\documentclass[letterpaper,UKenglish,numberwithinsect]{lipics-v2016}
\usepackage{microtype}%if unwanted, comment out or use option "draft"
\usepackage{algpseudocode,algorithm,epstopdf}
\usepackage{amsthm,amsmath,amssymb,amsfonts,cases,hyperref}
\usepackage{graphicx,subcaption}
\usepackage[capitalise,noabbrev]{cleveref}

\numberwithin{figure}{section}

\title{Approximation algorithms for two-machine flow-shop scheduling with a conflict graph%
\footnote{This work was partially supported by NSERC Canada and NSF China.}}
\titlerunning{Approximating $F2 \mid G = (V, E), p_{ij} \mid C_{\max}$ (v: \today)} %optional, in case that the title is too long; the running title should fit into the top page column

\author[1]{Yinhui Cai\footnote{Co-first authors.}}
\author[2]{Guangting Chen}
\author[3]{Yong Chen\footnote{Correspondence authors.}}
\author[4]{Randy Goebel}
\author[4]{Guohui Lin$^{\ddagger}$}
\author[4,5]{Longcheng Liu$^{\dagger}$}
\author[3]{An Zhang}
\affil[1]{School of Sciences, Hangzhou Dianzi University.  Hangzhou, Zhejiang 310018, China.}
\affil[2]{Taizhou University.  Taizhou, Zhejiang 317000, China.
  \texttt{gtchen@hdu.edu.cn}}
\affil[3]{Department of Mathematics, Hangzhou Dianzi University.  Hangzhou, Zhejiang 310018, China.
  \texttt{\{chenyong,anzhang\}@hdu.edu.cn}}
\affil[4]{Department of Computing Science, University of Alberta.  Edmonton, Alberta T6G 2E8, Canada.
  \texttt{\{rgoebel,guohui\}@ualberta.ca}}
\affil[5]{School of Mathematical Sciences, Xiamen University.  Xiamen, Fujian 361005, China.
  \texttt{longchengliu@xmu.edu.cn}}
\authorrunning{Y. Cai {\it et al.}} %mandatory. First: Use abbreviated first/middle names. Second (only in severe cases): Use first author plus 'et al..'

\Copyright{Yinhui Cai, Guangting Chen, Yong Chen, Randy Goebel, Guohui Lin, Longcheng Liu, An Zhang}%mandatory, please use full first names. LIPIcs license is "CC-BY";  http://creativecommons.org/licenses/by/3.0/

\subjclass{F.2.2 Sequencing and Scheduling; G.2.1 Combinatorial algorithms; G.4 Algorithm design and analysis}
\keywords{Flow-shop scheduling, conflict graph, $b$-matching, path cover, approximation algorithm}% mandatory: Please provide 1-5 keywords
\begin{document}

\maketitle

\begin{abstract}
%==================================================================================================
Path cover is a well-known intractable problem that finds a minimum number of vertex disjoint paths in a given graph to cover all the vertices.
We show that a variant, where the objective function is not the number of paths but the number of length-$0$ paths (that is, isolated vertices),
turns out to be polynomial-time solvable.
We further show that another variant, where the objective function is the total number of length-$0$ and length-$1$ paths, is also polynomial-time solvable.
Both variants find applications in approximating the two-machine flow-shop scheduling problem
in which job processing has constraints that are formulated as a conflict graph.
For the unit jobs, we present a $4/3$-approximation algorithm for the scheduling problem with an arbitrary conflict graph,
based on the exact algorithm for the variants of the path cover problem.
For the arbitrary jobs while the conflict graph is the union of two disjoint cliques, that is,
all the jobs can be partitioned into two groups such that the jobs in a group are pairwise conflicting,
we present a simple $3/2$-approximation algorithm.
%==================================================================================================
\end{abstract}

\section{Introduction}
%==================================================================================================
Scheduling is a well established research area that finds numerous applications in modern manufacturing industry and in operations research at large.
All scheduling problems modeling real-life applications have at least two components, the machines and the jobs.
In one big category of problems that have received intensive studies,
scheduling constraints are imposed between a machine and a job,
such as a time interval during which the job is allowed to be processed {\em nonpreemptively} on the machine,
while the machines are considered as independent from each other, so are the jobs.
For example, the parallel machine scheduling (the {\em multiprocessor scheduling} in \cite{GJ79}) is one of the first studied problems,
denoted as $Pm \mid \mid C_{\max}$ in the three-field notation~\cite{GLL79},
in which a set of jobs each needs to be processed by one of the $m$ given identical machines,
with the goal to minimize the maximum job completion time (called the {\em makespan});
the $m$-machine flow-shop scheduling (the {\em flow-shop scheduling} in \cite{GJ79}) is another first-studied problem,
denoted as $Fm \mid \mid C_{\max}$,
in which a set of jobs each needs to be processed by all the $m$ given machines in the same sequential order, with the goal to minimize the makespan.

In another category of scheduling problems, additional but limited resources are required for the machines to process the jobs~\cite{GG75}.
The resources are renewable but normally non-sharable in practice;
the jobs competing for the same resource have to be processed at different time if their total demand for a certain resource exceeds the supply.
Scheduling with resource constraints~\cite{GG75,GJ75} or scheduling with conflicts ({\sc SwC})~\cite{EHK09}
also finds numerous applications~\cite{BJ95,BC96,HKP03} and has attracted as much attention as the non-constrained counterpart.
In this paper, we use {\sc SwC} to refer to the nonpreemptive scheduling problems with additional constraints or conflicting relationships among the jobs
to disallow them to be processed concurrently on different machines.
We remark that in the literature, {\sc SwC} is also presented as the scheduling with agreements ({\sc SwA}),
in which a subset of jobs can be processed concurrently on different machines if and only if they are agreeing with each other~\cite{BB12,BB16}.
While in the most general scenario a conflict could involve multiple jobs,
in this paper we consider only those conflicts each involves two jobs and
consequently all the conflicts under consideration can be presented as a {\em conflict graph} $G = (V, E)$,
where $V$ is the set of jobs and an edge $e = (J_{j_1}, J_{j_2}) \in E$ represents a conflicting pair such that the two jobs $J_{j_1}$ and $J_{j_2}$ cannot be processed concurrently on different machines in any feasible schedule.

Extending the three-field notation \cite{GLL79},
the parallel machine {\sc SwC} with a conflict graph $G = (V, E)$ (also abbreviated as {\sc SCI} in the literature)~\cite{EHK09}
is denoted as $Pm \mid G = (V, E), p_j \mid C_{\max}$,
where the first field $Pm$ tells that there are $m$ parallel identical machines,
the second field describes the conflict graph $G = (V, E)$ over the set $V$ of all the jobs,
where the job $J_j$ requires a non-preemptive processing time of $p_j$ on any machine,
and the last field specifies the objective function to minimize the makespan $C_{\max}$.
One clearly sees when $E = \emptyset$, $Pm \mid G = (V, E), p_j \mid C_{\max}$ reduces to the classical multiprocessor scheduling $Pm \mid \mid C_{\max}$,
which is already NP-hard for $m \ge 2$~\cite{GJ79}.
Indeed, with $m$ either a given constant or part of input, $Pm \mid G = (V, E), p_j \mid C_{\max}$ is more difficult to approximate,
and there is a line of rich research to consider the unit jobs (that is, $p_j = 1$) and/or to consider certain special classes of conflict graphs.
The interested reader might see \cite{EHK09} and the references therein.

In this paper, we are interested in approximating the two-machine flow-shop {\sc SwC}.

In the general $m$-machine (also called $m$-stage) flow-shop~\cite{GJ79} denoted as $Fm \mid \mid C_{\max}$,
there are $m \ge 2$ machines $M_1, M_2, \ldots, M_m$,
a set $V$ of jobs each job $J_j$ needs to be processed through $M_1, M_2, \ldots, M_m$ sequentially 
with processing times $p_{1j}, p_{2j}, \ldots, p_{mj}$ respectively.
When $m = 2$, the two-machine flow-shop problem is polynomial time solvable, by Johnson's algorithm~\cite{Joh54};
the $m$-machine flow-shop problem when $m \ge 3$ is {\em strongly} NP-hard~\cite{GJS76}.
After several efforts~\cite{Joh54,GJS76,GS78,CGP96},
Hall presented a polynomial-time approximation scheme (PTAS) for the $m$-machine flow-shop problem, for any fixed integer $m \ge 3$~\cite{Hall98}.
When $m$ is part of input ({\it i.e.} an arbitrary integer), there is no known constant ratio approximation algorithm,
and the problem cannot be approximated within $1.25$ unless {\sc P = NP}~\cite{WHH97}.

The $m$-machine flow-shop {\sc SwC} was first studied in 1980's.
Blazewicz et al.~\cite{BLR83} considered multiple resource characteristics including the number of resource types,
resource availabilities and resource requirements;
they expanded the middle field of the three-field notation to express these resource characteristics,
for which the conflict relationships are modeled by complex structures such as hypergraphs.
At the end, they proved complexity results for several variants in which either the conflict relationships are simple enough or only the unit jobs are considered.
Further studies on more variants can be found in \cite{Roc83,Roc84,BCS86,BKS88,SKE92}.
In this paper, we consider those conflicts each involves only two jobs such that
all the conflicts under consideration can be presented as a conflict graph $G = (V, E)$.
The $m$-machine flow-shop scheduling with a conflict graph $G = (V, E)$ is denoted as $Fm \mid G = (V, E), p_{ij} \mid C_{\max}$.
We remark that our notation is slightly different from the one introduced by Blazewicz et al.~\cite{BLR83},
which uses a prefix ``res'' in the middle field for describing the resource characteristics.

Several applications of the $m$-machine flow-shop scheduling with a conflict graph were mentioned in the literature.
In a typical example of scheduling medical tests in an outpatient health care facility
where each patient (the job) needs to do a sequence of $m$ tests (the machines),
a patient must be accompanied by their doctor during a test and thus two patients under the care of the same doctor cannot go for tests simultaneously.
That is, two patients of the same doctor are conflicting to each other, and all the conflicts can be effectively described as a graph $G = (V, E)$,
where $V$ is the set of all the patients and an edge represents a conflicting pair of patients.

In two recent papers~\cite{TB17,TB18}, Tellache and Boudhar studied the problem $F2 \mid G = (V, E), p_{ij} \mid C_{\max}$, which they denote as {\sc FSC}.
In \cite{TB18}, the authors summarized and/or proved several complexity results;
to name a few, 
$F2 \mid G = (V, E), p_{ij} \mid C_{\max}$ is strongly NP-hard when $G = (V, E)$ is the complement of a complete split graph~\cite{TB18,BLR83}
(that is, $G$ is the union of a clique and an independent set),
$F2 \mid G = (V, E), p_{ij} \mid C_{\max}$ is weakly NP-hard when $G = (V, E)$ is the complement of a complete bipartite graph~\cite{TB18}
(that is, $G$ is the union of two disjoint cliques),
and for an arbitrary conflict graph $G = (V, E)$, $F2 \mid G = (V, E), p_{ij} = 1 \mid C_{\max}$ is strongly NP-hard~\cite{TB18}.
In \cite{TB17}, the authors proposed three mixed-integer linear programming models and a branch and bound algorithm to solve the last variant
$F2 \mid G = (V, E), p_{ij} = 1 \mid C_{\max}$ exactly;
their empirical study shows that the branch and bound algorithm outperforms and can solve instances of up to $20,000$ jobs.

In this paper, we pursue approximation algorithms with provable performance for the NP-hard variants of
the two-machine flow-shop scheduling with a conflict graph.
In Section~2, we present a $4/3$-approximation for the strongly NP-hard scheduling problem $F2 \mid G = (V, E), p_{ij} = 1 \mid C_{\max}$
for the unit jobs with an arbitrary conflict graph.
In Section~3, we present a simple $3/2$-approximation for the weakly NP-hard scheduling problem $F2 \mid G = K_\ell \cup K_{n-\ell}, p_{ij} \mid C_{\max}$
for arbitrary jobs with a conflict graph that is the union of two disjoint cliques (that is, the complement of a complete bipartite graph).
Some concluding remarks are provided in Section 4.

\section{Approximating $F2 \mid G = (V, E), p_{ij} = 1 \mid C_{\max}$}
%==================================================================================================
Tellache and Boudhar proved that $F2 \mid G = (V, E), p_{ij} = 1 \mid C_{\max}$ is strongly NP-hard by a reduction from the well known Hamiltonian path problem,
which is strongly NP-complete~\cite{GJ79}.
Furthermore, they remarked that $F2 \mid G = (V, E), p_{ij} = 1 \mid C_{\max}$ has a feasible schedule of makespan $C_{\max} = n + k$
if and only if the complement $\overline{G}$ of the conflict graph $G$, called the agreement graph, has a path cover of size $k$
(that is, a collection of $k$ vertex-disjoint paths that covers all the vertices of the graph $\overline{G}$),
where $n$ is the number of jobs (or vertices) in the instance.
This way, $F2 \mid G = (V, E), p_{ij} = 1 \mid C_{\max}$ is polynomially equivalent to the {\sc Path Cover} problem,
which is NP-hard even on some special classes of graphs including planar graphs~\cite{GJT76}, bipartite graphs~\cite{Gol04}, chordal graphs~\cite{Gol04},
chordal bipartite graphs~\cite{Mul96} and strongly chordal graphs~\cite{Mul96}.
In terms of approximability, to the best of our knowledge there is no $o(n)$-approximation for the Path Cover problem.

We give some terminologies first.
The conflict graphs considered in this paper are all simple graphs.
All paths and cycles in a graph are also simple.
The number of edges on a path/cycle defines the length of the path/cycle.
A length-$k$ path/cycle is also called a $k$-path/cycle for short.
Note that a single vertex is regarded as a $0$-path, while a cycle has length at least $3$.
For an integer $b \ge 1$, a $b$-matching of a graph is a spanning subgraph in which every vertex has degree no greater than $b$;
a maximum $b$-matching is a $b$-matching that contains the maximum number of edges.
A maximum $b$-matching of a graph can be computed in $O(m^2 \log n \log b)$-time,
where $n$ and $m$ are the number of vertices and the number of edges in the graph, respectively~\cite{Gab83}.
Clearly, a graph could have multiple distinct maximum $b$-matchings.

Given a graph, a path cover is a collection of vertex-disjoint paths in the graph that covers all the vertices,
and the size of the path cover is the number of paths therein.
The {\sc Path Cover} problem is to find a path cover of a given graph of the minimum size,
and the well known Hamiltonian path problem is to decide whether a given graph has a path cover of size $1$.
Besides the {\sc Path Cover} problem, many its variants have also been studied in the literature~\cite{AN07,PH08,AN10,RTM14}.
We mentioned earlier that Tellache and Boudhar proved that $F2 \mid G = (V, E), p_{ij} = 1 \mid C_{\max}$ is polynomially equivalent to the {\sc Path Cover} problem,
but to the best of our knowledge there is no approximation algorithm designed for $F2 \mid G = (V, E), p_{ij} = 1 \mid C_{\max}$.
Nevertheless, one easily sees that,
since $F2 \mid G = (V, E), p_{ij} = 1 \mid C_{\max}$ has a feasible schedule of makespan $C_{\max} = n + k$
if and only if the complement $\overline{G}$ of the conflict graph $G$ has a path cover of size $k$,
a trivial algorithm simply processing the jobs one by one (each on the first machine $M_1$ and then on the second machine $M_2$)
produces a schedule of makespan $C_{\max} = 2n$,
and thus is a $2$-approximation algorithm.

In this section, we will design two approximation algorithms with improved performance ratios for $F2 \mid G = (V, E), p_{ij} = 1 \mid C_{\max}$.
These two approximation algorithms are based on our polynomial time exact algorithms for two variants of the {\sc Path Cover} problem, respectively.
We start with the first variant called the {\em Path Cover with the minimum number of $0$-paths},
in which we are given a graph and we aim to find a path cover that contains the minimum number of $0$-paths.
In the second variant called the {\em Path Cover with the minimum number of $\{0, 1\}$-paths},
we aim to find a path cover that contains the minimum total number of $0$-paths and $1$-paths.
We remark that in both variants, we do not care about the size of the path cover.

\subsection{Path Cover with the minimum number of $0$-paths}
%--------------------------------------------------------------------------------------------------
Recall that in this variant of the {\sc Path Cover} problem, given a graph, we aim to find a path cover that contains the minimum number of $0$-paths.
The given graph is the complement $\overline{G} = (V, \overline{E})$ of the conflict graph $G = (V, E)$ in $F2 \mid G = (V, E), p_{ij} = 1 \mid C_{\max}$.
We next present a polynomial time algorithm that finds for $\overline{G}$ a path cover that contains the minimum number of $0$-paths.

In the first step, we apply any polynomial time algorithm to find a maximum $2$-matching in $\overline{G}$, denoted as $M$;
recall that this can be done in $O(m^2 \log n)$-time, where $n = |V|$ and $m = |\overline{E}|$.
$M$ is a collection of vertex-disjoint paths and cycles;
let ${\cal P}_0$ (${\cal P}_1$, ${\cal P}_2$, ${\cal P}_{\ge 3}$, ${\cal C}$, respectively) denote the sub-collection of $0$-paths
($1$-paths, $2$-paths, paths of length at least $3$, cycles, respectively) in $M$.
That is, $M = {\cal P}_0 \cup {\cal P}_1 \cup {\cal P}_2 \cup {\cal P}_{\ge 3} \cup {\cal C}$.

Clearly, if ${\cal P}_0 = \emptyset$, then we have a path cover containing no $0$-paths after removing one edge per cycle in ${\cal C}$.
In the following discussion we assume the existence of a $0$-path, which is often called a {\em singleton}.
We also call an ending vertex of a $k$-path with $k \ge 1$ as an endpoint for simplicity.
The following lemma is trivial due to the edge maximality of $M$.

\begin{lemma}
\label{lemma21}
All the singletons and endpoints in the maximum $2$-matching $M$ are pairwise non-adjacent to each other in the underlying graph $\overline{G}$.
\end{lemma}

Let $v_0$ be a singleton.
If $v_0$ is adjacent to a vertex $v_1$ on a cycle of ${\cal C}$ in the underlying graph $\overline{G}$,
then we may delete a cycle-edge incident at $v_1$ from $M$ while add the edge $(v_0, v_1)$ to $M$ to achieve another maximum $2$-matching with one less singleton.
Similarly, if $v_0$ is adjacent to a vertex $v_1$ on a path of ${\cal P}_{\ge 3}$ (note that $v_1$ has to be an internal vertex on the path by Lemma~\ref{lemma21})
in the underlying graph $\overline{G}$,
then we may delete a certain path-edge incident at $v_1$ from $M$ while add the edge $(v_0, v_1)$ to $M$
to achieve another maximum $2$-matching with one less singleton.
In either of the above two cases, assume the edge deleted from $M$ is $(v_1, v_2)$;
then we say the {\em alternating} path $v_0$-$v_1$-$v_2$ {\em saves} the singleton $v_0$.

In the general setting, in the underlying graph $\overline{G}$,
$v_0$ is adjacent to the middle vertex $v_1$ of a $2$-path $P_1$,
one endpoint $v_2$ of $P_1$ is adjacent to the middle vertex $v_3$ of another $2$-path $P_2$,
one endpoint $v_4$ of $P_2$ is adjacent to the middle vertex $v_5$ of another $2$-path $P_3$,
and so on,
one endpoint $v_{2i-2}$ of $P_{i-1}$ is adjacent to the middle vertex $v_{2i-1}$ of another $2$-path $P_i$,
one endpoint $v_{2i}$ of $P_i$ is adjacent to a vertex $v_{2i+1}$ of a cycle of ${\cal C}$ or a path of ${\cal P}_{\ge 3}$ (see an illustration in Figure~\ref{fig01}),
on which the edge $(v_{2i+1}, v_{2i+2})$ is to be deleted.
Then we may delete the edges $\{(v_{2j+1}, v_{2j+2}) \mid j = 0, 1, \ldots, i\}$ from $M$ while add the edges $\{(v_{2j}, v_{2j+1}) \mid j = 0, 1, \ldots, i\}$ to $M$
to achieve another maximum $2$-matching with one less singleton;
and we say the {\em alternating} path $v_0$-$v_1$-$v_2$-$\ldots$-$v_{2i}$-$v_{2i+1}$-$v_{2i+2}$ {\em saves} the singleton $v_0$.

\begin{figure}[h]
\centering
  \setlength{\unitlength}{0.7bp}%
  \begin{picture}(363.53, 164.44)(0,0)
  \put(0,0){\includegraphics[scale=0.7]{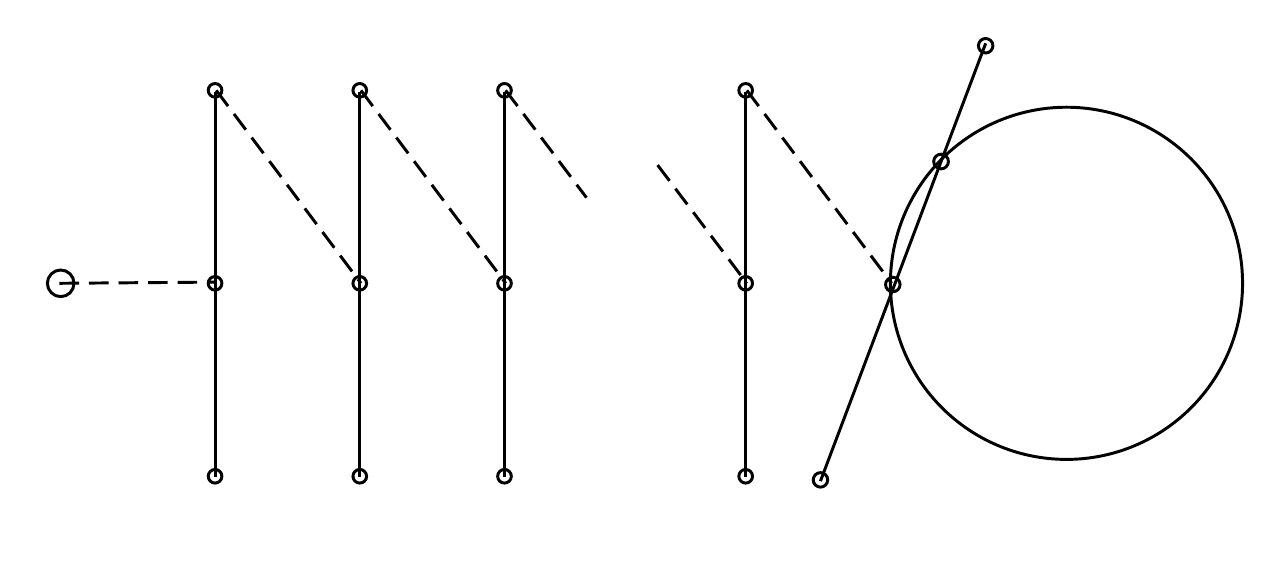}}
  \put(5.67,64.30){\fontsize{14.23}{17.07}\selectfont $v_0$}
  \put(61.24,64.30){\fontsize{14.23}{17.07}\selectfont $v_1$}
  \put(54.29,147.66){\fontsize{14.23}{17.07}\selectfont $v_2$}
  \put(102.92,64.30){\fontsize{14.23}{17.07}\selectfont $v_3$}
  \put(144.60,64.30){\fontsize{14.23}{17.07}\selectfont $v_5$}
  \put(214.06,64.30){\fontsize{14.23}{17.07}\selectfont $v_{2i-1}$}
  \put(262.69,78.20){\fontsize{14.23}{17.07}\selectfont $v_{2i+1}$}
  \put(95.97,147.66){\fontsize{14.23}{17.07}\selectfont $v_4$}
  \put(137.65,147.66){\fontsize{14.23}{17.07}\selectfont $v_6$}
  \put(207.12,147.66){\fontsize{14.23}{17.07}\selectfont $v_{2i}$}
  \put(269.64,105.98){\fontsize{14.23}{17.07}\selectfont $v_{2i+2}$}
  \put(54.29,8.73){\fontsize{14.23}{17.07}\selectfont $P_1$}
  \put(95.97,8.73){\fontsize{14.23}{17.07}\selectfont $P_2$}
  \put(137.65,8.73){\fontsize{14.23}{17.07}\selectfont $P_3$}
  \put(207.12,8.73){\fontsize{14.23}{17.07}\selectfont $P_i$}
  \put(172.39,8.73){\fontsize{14.23}{17.07}\selectfont $\ldots$}
  \put(165.44,85.14){\fontsize{14.23}{17.07}\selectfont $\ldots$}
  \put(172.39,147.66){\fontsize{14.23}{17.07}\selectfont $\ldots$}
  \end{picture}%
\caption{An alternating path $v_0$-$v_1$-$v_2$-$\ldots$-$v_{2i}$-$v_{2i+1}$-$v_{2i+2}$ that saves the singleton $v_0$,
	where the last two vertices are on a cycle of ${\cal C}$ or a path of ${\cal P}_{\ge 3}$.
	In the figure, solid edges are in the maximum $2$-matching $M$ and dashed edges are outside of $M$.\label{fig01}}
\end{figure}
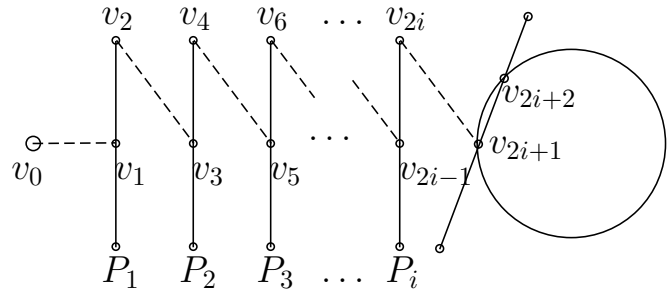

\begin{lemma}
\label{lemma22}
Given a maximum $2$-matching $M$ and a singleton $v_0$ therein,
finding a simple alternating path to save $v_0$, if exists, can be done in $O(m)$ time, where $m = |\overline{E}|$.
\end{lemma}
\begin{proof}
Firstly, if an alternating path is not simple, then a cycle that forms a subpath is also alternating and has an even length,
and thus the cycle can be removed resulting in a shorter alternating path.
Repeating this process if necessary, at the end we achieve a simple alternating path.
Therefore, we can limit the search for a simple alternating path.

Note that the edges on all possible alternating paths that save $v_0$ can be of the following four kinds:
1) all those edges incident at $v_0$, each oriented out of $v_0$;
2) all those edges of the $2$-paths, each oriented from the middle vertex to the endpoint;
3) all those edges each connecting an endpoint of a $2$-path to the middle vertex of another $2$-path, oriented from the endpoint to the middle vertex;
4) all those edges each connecting an endpoint of a $2$-path to a vertex on some path of ${\cal P}_{\ge 3}$ or on some cycle of ${\cal C}$,
	oriented out of the endpoint.
If follows that by a BFS ({\em breadth-first search}) traversal starting from $v_0$ in the digraph formed by the above four kinds of oriented edges,
if a vertex on some path of ${\cal P}_{\ge 3}$ or on some cycle of ${\cal C}$ can be reached then we achieve a simple alternating path;
otherwise, we conclude that no alternating path saving the singleton $v_0$ exists.
Both construction of the digraph and the BFS traversal take $O(m)$ time.
This proves the lemma.
\end{proof}

The second step of the algorithm is to iteratively find a simple alternating path to save a singleton;
it terminates when no alternating path is found.
The resulting maximum $2$-matching is still denoted as $M$.

In the last step, we break the cycles in $M$ by deleting one edge per cycle to produce a path cover.
Denote our algorithm as {\sc Algorithm A}, of which a high-level description is provided in Figure~\ref{fig02}.
We will prove in the next theorem that the path cover produced by {\sc Algorithm A} contains the minimum number of $0$-paths.

\begin{figure}[h]
\begin{center}
\framebox{
\begin{minipage}{4.5in}
{\sc Algorithm A}($\overline{G} = (V, \overline{E})$):
\begin{description}
\parskip=0pt
\item[Step 1.]
	Compute a maximum $2$-matching $M$;
\item[Step 2.]
	repeatedly find an alternating path to save a singleton in $M$,\\
	till either no singleton exists or no alternating path is found;
\item[Step 3.]
	break cycles in $M$ by removing one edge per cycle, and\\
	return the resulting path cover.
\end{description}
\end{minipage}}
\end{center}
\caption{A high-level description of {\sc Algorithm A} for computing a path cover in the agreement graph $\overline{G} = (V, \overline{E})$.\label{fig02}}
\end{figure}

\begin{theorem}
\label{thm23}
{\sc Algorithm A} is an $O(m^2 \log n)$-time algorithm for computing a path cover with the minimum number of $0$-paths
in the agreement graph $\overline{G} = (V, \overline{E})$.
\end{theorem}
\begin{proof}
Recall that the last step of {\sc Algorithm A} is to break cycles only.
We thus use the maximum $2$-matching achieved at the end of the second step in the following proof, denoted as $M$.
We point out that in the second step, in each iteration where an alternating path is found to save a singleton of the current maximum $2$-matching,
we swap the edges on the alternating path inside the matching with the edges outside of the matching to
{\em move} from the current maximum $2$-matching to another maximum $2$-matching which contains one less singleton.

We prove the theorem by the minimal counterexample.

Recall that ${\cal P}_0$ contains all the singletons (that is, $0$-paths) in $M$.
Let $M^*$ be an optimal path cover that contains the minimum number of singletons, and let ${\cal P}_0^*$ denote this collection of singletons.
Assume to the contrary that the path cover obtained from $M$ contains more than the minimum number of singletons, then we must have
\begin{equation}
\label{eq1}
|{\cal P}_0| > |{\cal P}_0^*| \ge 0.
\end{equation}
Assume our agreement graph $\overline{G} = (V, \overline{E})$ is a minimal graph on which Eq.~(\ref{eq1}) holds,
then $M$ and $M^*$ should not have any common singleton as otherwise it can be deleted to obtain a smaller graph.
That is,
\begin{equation}
\label{eq2}
{\cal P}_0 \cap {\cal P}_0^* = \emptyset.
\end{equation}

It follows that a singleton $v_0 \in {\cal P}_0$ is not a singleton in $M^*$.
Suppose $(v_0, v_1) \in M^*$.
From the edge maximality of $M$ and the non-existence of an alternating path, we conclude that $v_1$ has to be the middle vertex of some $2$-path $P_1 \in {\cal P}_2$.

Let $u_1$ and $v_2$ be the two endpoints of the $2$-path $P_1$.
For the same reason as for the singleton $v_0$, from the edge maximality of $M$ and the non-existence of an alternating path,
we conclude that in $\overline{G}$ each of $u_1$ and $v_2$ can be adjacent to only the middle vertices of $2$-paths, including $v_1$.
On the other hand, we conclude that none of $u_1$ and $v_2$ can be a singleton in $M^*$.
We prove this by contradiction to assume for instance $v_2$ is a singleton in $M^*$;
then the alternating path $v_0$-$v_1$-$v_2$ would save $v_0$ but leave $v_2$ as a new singleton,
which subsequently gives rise to another maximum $2$-matching $M'$ with the same number of singletons but $M'$ shares with $M^*$ a common singleton $v_2$,
a contradiction to the minimality of $\overline{G}$.

The last paragraph essentially implies that in $M^*$, each of $u_1$ and $v_2$ is adjacent to the middle vertex of a certain $2$-path of ${\cal P}_2$.
Since the edges $(u_1, v_1)$ and $(v_1, v_2)$ cannot both be in $M^*$ (otherwise $v_1$ would have degree $3$ in $M^*$),
in $M^*$ one of $u_1$ and $v_2$, and assume without loss of generality $v_2$,
is adjacent to the middle vertex $v_3$ of a $2$-path $P_2 \in {\cal P}_2$ other than $P_1$.

Let $u_2$ and $v_4$ be the two endpoints of the second $2$-path $P_2$.
For the same reason as for the singleton $v_0$, from the edge maximality of $M$ and the non-existence of an alternating path,
we conclude that each of $u_2$ and $v_4$ can be adjacent to only the middle vertices of $2$-paths, including $v_1$ and $v_3$.
On the other hand, we conclude that none of $u_2$ and $v_4$ can be a singleton in $M^*$, for the same reason as for $v_2$ in the above.
These imply that in $M^*$, each of $u_2$ and $v_4$ is adjacent to the middle vertex of a certain $2$-path of ${\cal P}_2$.
Since there are four endpoints $\{u_1, v_2, u_2, v_4\}$ but only two middle vertices $\{v_1, v_3\}$ for $P_1$ and $P_2$,
in $M^*$ one of the four endpoints $\{u_1, v_2, u_2, v_4\}$, and assume without loss of generality $v_4$,
is adjacent to the middle vertex $v_5$ of a $2$-path $P_3 \in {\cal P}_2$ other than $P_1$ and $P_2$.

Let $u_3$ and $v_6$ be the two endpoints of the third $2$-path $P_3$.
Repeat the same argument as before we have that $v_6$ is adjacent to the middle vertex $v_7$ of a fourth $2$-path $P_4 \in {\cal P}_2$.
And so on.
These contradict the fact that the graph $\overline{G}$ is finite,
and therefore the maximum $2$-matching at the end of the second step of {\sc Algorithm A} contains the minimum number of singletons.

For the running time, since in each iteration of the second step we may ``{\em glue}'' all singletons as one for finding an alternating path.
If no alternating path is found, then the second step terminates;
otherwise one can easily check which singletons are the root of the alternating path and pick to save one of them, and the iteration ends.
It follows that there could be $O(n)$ iterations and each iteration needs $O(m)$ time, and thus the total running time for the second step is $O(nm)$.
Clearly the last step can be done in $O(n)$ time.
Therefore the running time of {\sc Algorithm A} is dominated by the first step of finding a maximum $2$-matching, which is done in $O(m^2 \log n)$ time.
This finishes the proof of the theorem.
\end{proof}

\subsection{Path Cover with the minimum number of $\{0, 1\}$-paths}
%--------------------------------------------------------------------------------------------------
In this variant of the {\sc Path Cover} problem, given a graph, we aim to find a path cover that contains the minimum total number of $0$-paths and $1$-paths.
Again, the given graph is the complement $\overline{G} = (V, \overline{E})$ of the conflict graph $G = (V, E)$ in $F2 \mid G = (V, E), p_{ij} = 1 \mid C_{\max}$.
We next present a polynomial time algorithm called {\sc Algorithm B} that finds for $\overline{G}$ such a path cover.

Recall that in {\sc Algorithm A} for computing a path cover that contains the minimum number of $0$-paths,
an alternating path saving a singleton $v_0$ starts from the singleton $v_0$ and
reaches a vertex $v_{2i+1}$ on a path of ${\cal P}_{\ge 3}$ or on a cycle of ${\cal C}$ (see Figure~\ref{fig01}).
If $v_{2i+1}$ is on a cycle, then the last vertex $v_{2i+2}$ can be any one of the two neighbors of $v_{2i+1}$ on the cycle.
If $v_{2i+1}$ is on a $k$-path, then the last vertex $v_{2i+2}$ is a {\em non-endpoint} neighbor of $v_{2i+1}$ on the path (the existence is guaranteed by $k \ge 3$);
and the reason why $v_{2i+2}$ cannot be an endpoint is obvious since otherwise $v_{2i+2}$ would be left as a new singleton after the edge swapping.
In the current variant we want to minimize the total number of $0$-paths and $1$-paths;
clearly $v_{2i+2}$ cannot be an endpoint either and cannot even be the vertex adjacent to an endpoint,
for the latter case because the edge swapping saves $v_0$ but leaves a new $1$-path.
To guarantee the existence of such vertex $v_{2i+2}$, the $k$-path must have $k \ge 4$, and if $k = 4$ then $v_{2i+1}$ cannot be the middle vertex of the $4$-path.

{\sc Algorithm B} is in spirit similar to but in practice slightly more complex than {\sc Algorithm A},
mostly because the definition of an alternating path saving a singleton or a $1$-path is different, and slightly more complex.

In the first step of {\sc Algorithm B}, we apply any polynomial time algorithm to find a maximum $2$-matching $M$ in $\overline{G}$.
Let ${\cal P}_0$ (${\cal P}_1$, ${\cal P}_2$, ${\cal P}_3$, ${\cal P}_4$, ${\cal P}_{\ge 5}$, ${\cal C}$, respectively) denote the sub-collection of $0$-paths
($1$-paths, $2$-paths, $3$-paths, $4$-paths, paths of length at least $5$, cycles, respectively) in $M$.
We also let ${\cal P}_{0,1} = {\cal P}_0 \cup {\cal P}_1$ denote the collection of all $0$-paths (called {\em singletons}) and $1$-paths in $M$.

Let $e_0 = (v_0, u_0)$ be an edge in $M$.
In the sequel when we say $e_0$ is {\em adjacent} to a vertex $v_1$ in the graph $\overline{G}$,
we mean $v_1$ is a different vertex (from $v_0$ and $u_0$) and at least one of $v_0$ and $u_0$ is adjacent to $v_1$;
if both $v_0$ and $u_0$ are adjacent to $v_1$, then pick one (often arbitrarily) for the subsequent purposes.
This way, we unify our treatment on singletons and $1$-paths, for the reasons to be seen in the following.
For ease of presentation, we use an {\em object} to refer to a vertex or an edge.
Like in the last subsection, an ending vertex of a $k$-path with $k \ge 1$ or an ending edge of a $k$-path with $k \ge 2$ is called an {\em end-object} for simplicity.

Let $v_0$ be a singleton or $e_0 = (v_0, u_0)$ be a $1$-path in $M$.
In the underlying graph $\overline{G}$,
if $v_0$ is adjacent to a vertex $v_1$ on a cycle of ${\cal C}$, or on a path of ${\cal P}_{\ge 5}$, or on a $4$-path such that $v_1$ is not the middle vertex,
then we may delete a certain edge incident at $v_1$ from $M$ while add the edge $(v_0, v_1)$ to $M$ to achieve another maximum $2$-matching
with one less singleton if $v_0$ is a singleton or with one less $1$-path.
In either of the three cases, assume the edge deleted from $M$ is $(v_1, v_2)$;
then we say the {\em alternating} path $v_0$-$v_1$-$v_2$ {\em saves} the singleton $v_0$ or the $1$-path $e_0 = (v_0, u_0)$.

\begin{figure}[h]
\centering
  \setlength{\unitlength}{0.7bp}%
  \begin{picture}(363.53, 183.82)(0,0)
  \put(0,0){\includegraphics[scale=0.7]{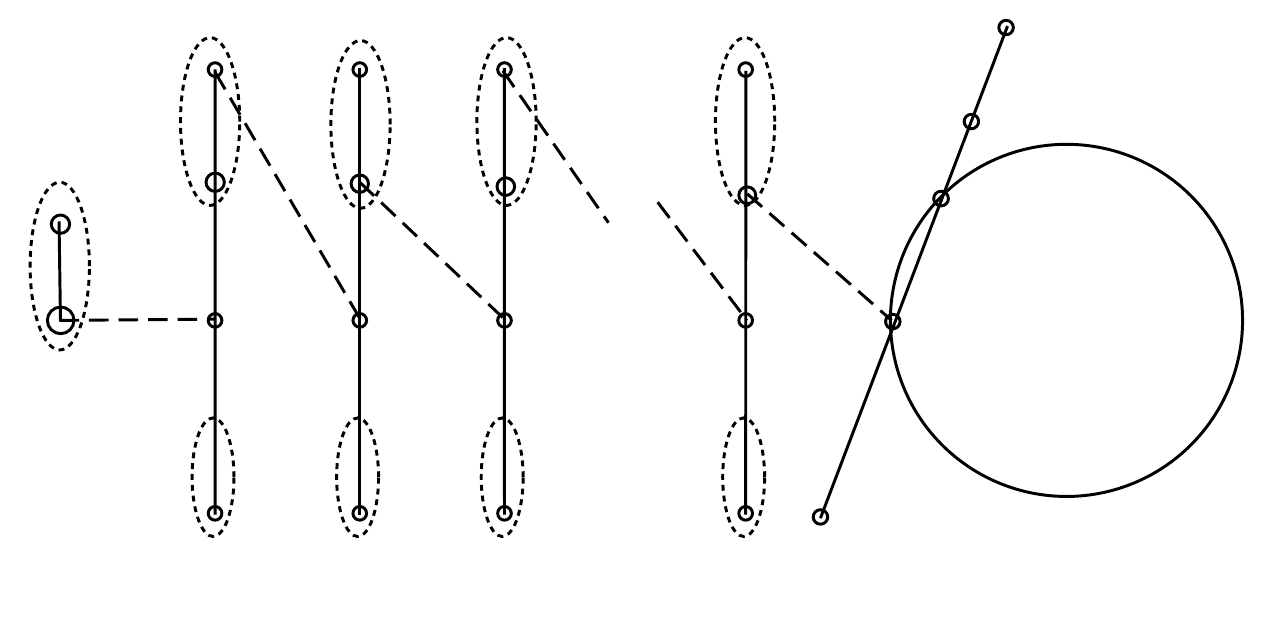}}
  \put(5.67,72.63){\fontsize{14.23}{17.07}\selectfont $v_0$}
  \put(61.24,72.63){\fontsize{14.23}{17.07}\selectfont $v_1$}
  \put(62.62,114.35){\fontsize{14.23}{17.07}\selectfont $v_2$}
  \put(102.92,72.63){\fontsize{14.23}{17.07}\selectfont $v_3$}
  \put(144.60,72.63){\fontsize{14.23}{17.07}\selectfont $v_5$}
  \put(214.06,72.63){\fontsize{14.23}{17.07}\selectfont $v_{2i-1}$}
  \put(262.69,86.52){\fontsize{14.23}{17.07}\selectfont $v_{2i+1}$}
  \put(104.30,114.35){\fontsize{14.23}{17.07}\selectfont $v_4$}
  \put(145.98,114.35){\fontsize{14.23}{17.07}\selectfont $v_6$}
  \put(215.45,114.35){\fontsize{14.23}{17.07}\selectfont $v_{2i}$}
  \put(269.64,114.31){\fontsize{14.23}{17.07}\selectfont $v_{2i+2}$}
  \put(54.29,8.73){\fontsize{14.23}{17.07}\selectfont $P_1$}
  \put(95.97,8.73){\fontsize{14.23}{17.07}\selectfont $P_2$}
  \put(137.65,8.73){\fontsize{14.23}{17.07}\selectfont $P_3$}
  \put(207.12,8.73){\fontsize{14.23}{17.07}\selectfont $P_i$}
  \put(172.39,8.73){\fontsize{14.23}{17.07}\selectfont $\ldots$}
  \put(165.44,93.47){\fontsize{14.23}{17.07}\selectfont $\ldots$}
  \put(172.39,114.35){\fontsize{14.23}{17.07}\selectfont $\ldots$}
  \end{picture}%
\caption{An alternating path $v_0$-$v_1$-$v_2$-$\ldots$-$v_{2i}$-$v_{2i+1}$-$v_{2i+2}$ that saves the singleton $v_0$ or the $1$-path $e_0 = (v_0, u_0)$,
	where the last two vertices are on a cycle of ${\cal C}$, or on a path of ${\cal P}_{\ge 5}$, or on a $4$-path such that $v_{2i+1}$ is not the middle vertex.
	In the figure, solid edges are in the maximum $2$-matching $M$, dashed edges are outside of $M$,
	and a dotted circle contains an object which is either a vertex or an edge.\label{fig03}}
\end{figure}
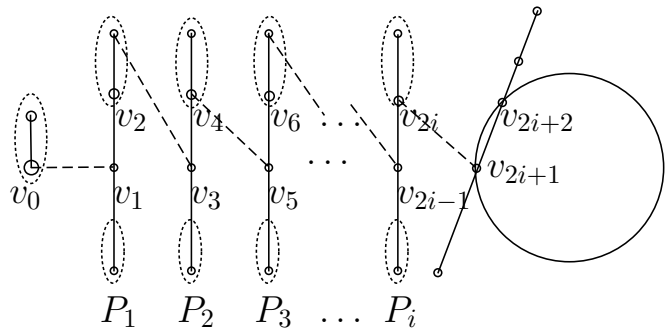

Analogously as in the last subsection, in the general setting, in the underlying graph $\overline{G}$,
$v_0$ is adjacent to a vertex $v_1$ of a path $P_1 \in {\cal P}_{2,3,4}$ (if $P_1$ is a $4$-path then $v_1$ has to be the middle vertex).
Note that this vertex $v_1$ basically separates the two end-objects of the path $P_1$ --- 
an analogue to the role of the middle vertex of a $2$-path that separates the two endpoints of the $2$-path.
We say ``an end-object of $P_1$ is adjacent to $v_1$ via $v_2$'', to mean that if the end-object is a vertex then it is $v_2$,
or if the end-object is an edge, then it is $(v_2, u_2)$, with the edge $(v_1, v_2)$ on the path $P_1$ either way (see an illustration in Figure~\ref{fig03}).

Suppose one end-object of $P_1$, which is adjacent to $v_1$ via $v_2$, is adjacent to a vertex $v_3$ of another $P_2 \in {\cal P}_{2,3,4}$
(the same, if $P_2$ is a $4$-path then $v_3$ has to be the middle vertex);
one end-object of $P_2$, which is adjacent to $v_3$ via $v_4$, is adjacent to a vertex $v_5$ of another $P_3 \in {\cal P}_{2,3,4}$
(the same, if $P_3$ is a $4$-path then $v_5$ has to be the middle vertex);
and so on;
one end-object of $P_{i-1}$, which is adjacent to $v_{2i-3}$ via $v_{2i-2}$, is adjacent to a vertex $v_{2i-1}$ of another $P_i \in {\cal P}_{2,3,4}$
(the same, if $P_i$ is a $4$-path then $v_{2i-1}$ has to be the middle vertex);
one end-object of $P_i$, which is adjacent to $v_{2i-1}$ via $v_{2i}$, is adjacent to a vertex $v_{2i+1}$ of a cycle of ${\cal C}$, or of a path of ${\cal P}_{\ge 5}$,
or of a $4$-path such that $v_{2i+1}$ is not the middle vertex (see an illustration in Figure~\ref{fig03}),
on which a certain edge $(v_{2i+1}, v_{2i+2})$ is to be deleted.
Then we may delete the edges $\{(v_{2j+1}, v_{2j+2}) \mid j = 0, 1, \ldots, i\}$ from $M$ while add the edges $\{(v_{2j}, v_{2j+1}) \mid j = 0, 1, \ldots, i\}$ to $M$
to achieve another maximum $2$-matching with one less singleton if $v_0$ is a singleton or with one less $1$-path.
We say the {\em alternating} path $v_0$-$v_1$-$v_2$-$\ldots$-$v_{2i}$-$v_{2i+1}$-$v_{2i+2}$ {\em saves} the singleton $v_0$ or the $1$-path $e_0 = (v_0, u_0)$.
It is important to note that in this alternating path, the vertex $v_2$ ``{\em represents}'' the end-object of $P_1$,
meaning that when the end-object is an edge, it is treated very the same as the vertex $v_2$.

\begin{lemma}
\label{lemma24}
Given a maximum $2$-matching $M$ and an object in ${\cal P}_{0,1}$,
finding a simple alternating path to save the object, if exists, can be done in $O(m)$ time, where $m = |\overline{E}|$.
\end{lemma}
\begin{proof}
The lemma is a generalization of Lemma~\ref{lemma22}.

Firstly, if an alternating path is not simple, then a cycle that forms a subpath is also alternating and has an even length,
and thus the cycle can be removed resulting a shorter alternating path.
Repeating this process if necessary, at the end we achieve a simple alternating path.
Therefore, we can limit the search for a simple alternating path.

Note that the edges on all possible alternating paths that save an object in ${\cal P}_{0,1}$ can be of the following four kinds:
1) all those edges incident at a vertex of the object, each oriented out of the vertex;
2) all those edges of the paths of ${\cal P}_{2,3,4}$, each oriented towards an endpoint and each internal edge is bidirected;
3) all those edges each connecting a vertex of an end-object of a path of ${\cal P}_{2,3,4}$ to an internal vertex of another such path,
	oriented from the vertex of the end-object;
	note that if the second path is a $4$-path, then the internal vertex of this $4$-path has to be the middle vertex;
4) all those edges each connecting a vertex of an end-object of a path of ${\cal P}_{2,3,4}$ to a vertex on some path of ${\cal P}_{\ge 5}$,
	or on some cycle of ${\cal C}$, or on some $4$-path such that the vertex is not the middle vertex, oriented out of the vertex of the end-object.

If follows that by a BFS traversal starting from the vertices of an object of ${\cal P}_{0,1}$
in the digraph formed by the above four kinds of oriented edges,
if a vertex on some path of ${\cal P}_{\ge 5}$, or on some cycle of ${\cal C}$, or on some $4$-path such that the vertex is not the middle vertex,
can be reached then we achieve a simple alternating path;
otherwise, we conclude that no alternating path saving the object exists.
Both construction of the digraph and the BFS traversal take $O(m)$ time.
This proves the lemma.
\end{proof}

The second step of the algorithm is to iteratively find a simple alternating path to save an object of ${\cal P}_{0,1}$;
it terminates when no alternating path is found.
The resulting maximum $2$-matching is still denoted as $M$.

In the last step, we break the cycles in $M$ by deleting one edge per cycle to produce a path cover.
A high-level description of {\sc Algorithm B} is similar to the one for {\sc Algorithm A} shown in Figure~\ref{fig02},
replacing a singleton by an object of ${\cal P}_{0,1}$.
We will prove in Theorem~\ref{thm25} that the path cover produced by {\sc Algorithm B} contains the minimum total number of $0$-paths and $1$-paths.

\begin{theorem}
\label{thm25}
{\sc Algorithm B} is an $O(m^2 \log n)$-time algorithm for computing a path cover with the minimum total number of $0$-paths and $1$-paths
in the agreement graph $\overline{G} = (V, \overline{E})$.
\end{theorem}
\begin{proof}
The proof is similar to the proof of Theorem~\ref{thm23}.

Recall that the last step of {\sc Algorithm B} is to break cycles only.
We thus use the maximum $2$-matching achieved at the end of the second step in the following proof, denoted as $M$.
We point out that in the second step, in each iteration where an alternating path is found to save an object of ${\cal P}_{0,1}$ of the current maximum $2$-matching,
we swap the edges on the alternating path inside in the matching with the edges outside of the matching
to {\em move} from the current maximum $2$-matching to another maximum $2$-matching
(which contains one less object of ${\cal P}_{0,1}$).

We prove the theorem by the minimal counterexample.

Let $M^*$ be an optimal path cover that contains the minimum total number of $0$-paths and $1$-paths,
and similarly let ${\cal P}_i^*$ denote the sub-collection of the $i$-paths in $M^*$, for $i = 0, 1, 2, \ldots$.
Assume to the contrary that
\begin{equation}
\label{eq3}
|{\cal P}_{0,1}| > |{\cal P}_{0,1}^*| \ge 0,
\end{equation}
and assume our agreement graph $\overline{G} = (V, \overline{E})$ is a minimal graph on which Eq.~(\ref{eq3}) holds,
then $M$ and $M^*$ should not have any common singleton or any common $1$-path as otherwise it can be deleted to obtain a smaller graph.
That is,
\begin{equation}
\label{eq4}
{\cal P}_0 \cap {\cal P}_0^* = \emptyset, \mbox{ and } {\cal P}_1 \cap {\cal P}_1^* = \emptyset.
\end{equation}

It follows that an object of ${\cal P}_{0,1}$ is not an object of ${\cal P}_{0,1}^*$.
In the sequel we assume there is a singleton $v_0 \in {\cal P}_0$ and show that it leads to a contradiction.
A similar contradiction can be constructed if there is a $1$-path in ${\cal P}_1$.
Since $v_0$ is not a singleton in $M^*$, we may suppose $(v_0, v_1) \in M^*$.
From the edge maximality of $M$ and the non-existence of an alternating path to save $v_0$,
we conclude that $v_1$ has to be a vertex of some path $P_1 \in {\cal P}_{2,3}$ or the middle vertex of some $4$-path $P_1$,
such that $v_1$ separates the two end-objects of $P_1$.

If an end-object is a vertex, denoted as $v_2$,
then the same from the edge maximality of $M$ and the non-existence of an alternating path to save $v_0$,
we conclude that $v_2$ behaves the same as $v_0$,
that it can be adjacent to only the vertices of paths in ${\cal P}_{2,3}$ or the middle vertices of $4$-paths, including $v_1$.
On the other hand, $v_2$ cannot be a singleton in $M^*$,
since otherwise the alternating path $v_0$-$v_1$-$v_2$ would save $v_0$ but leave $v_2$ as a new singleton,
which subsequently gives rise to another maximum $2$-matching $M'$ with the same number of singletons but $M'$ shares with $M^*$ a common singleton,
a contradiction to the minimality of $\overline{G}$.

If an end-object is an edge, denoted as $(v_2, u_2)$, adjacent to $v_1$ via $v_2$,
then the same from the edge maximality of $M$ and the non-existence of an alternating path to save $v_0$,
we conclude that both $v_2$ and $u_2$ behave the same as $v_0$,
that each can be adjacent to only the vertices of paths in ${\cal P}_{2,3}$ or the middle vertices of $4$-paths, including $v_1$.
On the other hand, at least one of $v_2$ and $u_2$ should be adjacent to a third vertex in $M^*$.
Assume this is not the case, then $(v_2, u_2)$ has to be a $1$-path in $M^*$ ($v_2$ and $u_2$ cannot both be singletons due to the existence of the edge $(v_2, u_2)$).
It follows that the alternating path $v_0$-$v_1$-$v_2$ would save $v_0$ but leave $(v_2, u_2)$ as a new $1$-path,
which subsequently gives rise to another maximum $2$-matching $M'$ with the same total number of singletons and $1$-paths
but $M'$ shares with $M^*$ a common $1$-path,
a contradiction to the minimality of $\overline{G}$.

Since the two path-edges that $v_1$ is incident to in $P_1$ cannot both be in $M^*$ (otherwise $v_1$ would have degree $3$ in $M^*$),
in $M^*$ one end-object of the path $P_1$ is adjacent to a vertex of some path in ${\cal P}_{2,3}$ or the middle vertex of some $4$-path denoted as $P_2$.
Let $v_2$ denote the vertex through which this end-object of the path $P_1$ connects to $v_1$, and $v_3$ denote the vertex on the path $P_2$.

Similarly as in the proof of Theorem~\ref{thm23}, we may repeat the above argument on the path $P_1$ for the new path $P_2$,
to either contradict the minimality of the graph $\overline{G}$ or introduce another new path $P_3$;
and so on.
The latter cases together contradict the fact that the graph $\overline{G}$ is finite,
and therefore the maximum $2$-matching at the end of the second step of {\sc Algorithm B} contains the minimum total number of singletons and $1$-paths.

For the running time, since in each iteration of the second step again we may ``{\em glue}'' all singletons and the endpoints of all the $1$-paths
as one for finding an alternating path.
If no alternating path is found, then the second step terminates;
otherwise one can easily check which singletons and/or $1$-paths are the root of the alternating path and pick to save one of them, and the iteration ends.
It follows that there could be $O(n)$ iterations and each iteration needs $O(m)$ time, and thus the total running time for the second step is $O(nm)$.
Clearly the last step can be done in $O(n)$ time.
Therefore the running time of {\sc Algorithm B} is dominated by the first step of finding a maximum $2$-matching, which is done in $O(m^2 \log n)$ time.
This finishes the proof of the theorem.
\end{proof}

\begin{remark}
\label{remark26}
The path cover produced by {\sc Algorithm B} has the minimum total number of $0$-paths and $1$-paths in the agreement graph $\overline{G} = (V, \overline{E})$.
One may run {\sc Algorithm A} at the end of the second step of {\sc Algorithm B} to achieve a path cover with the minimum total number of $0$-paths and $1$-paths,
and with the minimum number of $0$-paths.
During the execution of {\sc Algorithm A}, a singleton trades for a $1$-path.
\end{remark}

\subsection{Approximation algorithms for $F2 \mid G = (V, E), p_{ij} = 1 \mid C_{\max}$}
%--------------------------------------------------------------------------------------------------
Given an instance of the problem $F2 \mid G = (V, E), p_{ij} = 1 \mid C_{\max}$,
where there are $n$ unit jobs $V = \{J_1, J_2, \ldots, J_n\}$ to be processed on the two-machine flow-shop, with their conflict graph $G = (V, E)$,
we want to find a schedule with a provable makespan.

For a $k$-path in the agreement graph $\overline{G} = (V, \overline{E})$, where $k \ge 0$, for example $P = J_1$-$J_2$-$\ldots$-$J_k$-$J_{k+1}$,
we compose a sub-schedule $\pi_P$ in which the machine $M_1$ continuously processes the jobs $J_1, J_2, \ldots, J_{k+1}$ in order,
and the machine $M_2$ in one unit of time after $M_1$ continuously processes these jobs in the same order.
The sub-makespan for the flow-shop to complete these $k+1$ jobs is thus $k+2$ (units of time).
Let $M = \{P_1, P_2, \ldots, P_\ell\}$ be a path cover of size $\ell$ in the agreement graph $\overline{G}$.
For each path $P_i$ we use $|P_i|$ to denote its length and construct the sub-schedule $\pi_{P_i}$ as above that has a sub-makespan of $|P_i| + 2$.
We then concatenating these $\ell$ sub-schedules (in an arbitrary order) into a full schedule $\pi$, which clearly has a makespan
\begin{equation}
\label{eq5}
C_{\max}^{\pi} = \sum_{i=1}^\ell (|P_i| + 2) = n + \ell.
\end{equation}

On the other hand, given a schedule $\pi$, if two jobs $J_{j_1}$ and $J_{j_2}$ are processed concurrently on the two machines,
then they have to be agreeing to each other and thus adjacent in the agreement graph $\overline{G}$;
we select this edge $(J_{j_1}, J_{j_2})$.
Note that one job can be processed concurrently with at most two other jobs as there are only two machines.
Therefore, all the selected edges form into a number of vertex-disjoint paths in $\overline{G}$ (due to the flow-shop, no cycle is formed);
these paths together with the vertices outside of the paths, which are the $0$-paths, form a path cover for $\overline{G}$.
Assuming without loss of generality that two machines cannot both idle at any time point, the makespan of the schedule is exactly calculated as in Eq.~(\ref{eq5}).

We state this relationship between a feasible schedule and a path cover in the agreement graph $\overline{G}$ into the following lemma.

\begin{lemma}{\rm \cite{TB18}}
\label{lemma26}
A feasible schedule $\pi$ for the problem $F2 \mid G = (V, E), p_{ij} = 1 \mid C_{\max}$ one-to-one
corresponds to a path cover $M$ in the agreement graph $\overline{G}$,
and $C_{\max}^{\pi} = n + |M|$, where $n$ is the number of jobs in the instance.
\end{lemma}

\begin{theorem}
\label{thm27}
The problem $F2 \mid G = (V, E), p_{ij} = 1 \mid C_{\max}$ admits an $O(m^2 \log n)$-time $4/3$-approximation algorithm,
where $n = |V|$ and $m = |\overline{E}|$.
\end{theorem}
\begin{proof}
Let $\pi^*$ denote an optimal schedule for the problem $F2 \mid G = (V, E), p_{ij} = 1 \mid C_{\max}$ with a makespan $C_{\max}^*$,
and $M^*$ be the corresponding path cover in the agreement graph $\overline{G}$.
The sub-collection of $0$-paths and $1$-paths in $M^*$ is denoted as ${\cal P}^*_{0,1}$.

Let $M$ be the path cover computed by {\sc Algorithm B} for the agreement graph $\overline{G}$ that achieves the minimum total number of $0$-paths and $1$-paths.
The sub-collections of $0$-paths and $1$-paths in $M$ are denoted as ${\cal P}_0$ and ${\cal P}_1$, respectively, and ${\cal P}_{0,1}$ denotes their union.
Then we have
\[
|{\cal P}_{0,1}^*| \ge |{\cal P}_{0,1}|.
\]

From Lemma~\ref{lemma26}, we have
\[
C_{\max}^* = n + |M^*| \ge n + |{\cal P}_{0,1}^*| \ge n + |{\cal P}_{0,1}|.
\]
It follows also from Lemma~\ref{lemma26} that the schedule constructed using the path cover $M$ has a makespan
\[
C_{\max} = n + |M| \le n + |{\cal P}_{0,1}| + \frac 13 \left(n - |{\cal P}_0| - 2 |{\cal P}_1| \right) \le \frac 43 C_{\max}^*,
\]
since every path of length $2$ or above contains at least three vertices.

Note that the running time of the approximation algorithm is obvious,
which calls {\sc Algorithm B} and then constructs the schedule in $O(n)$ time using the computed path cover.
\end{proof}

\begin{remark}
\label{remark28}
If {\sc Algorithm A} is used in the proof of Theorem~\ref{thm27} to compute a path cover with the minimum number of $0$-paths and
subsequently to construct a schedule $\pi$,
then we have $C_{\max}^{\pi} \le \frac 32 C_{\max}^*$.
That is, we have an $O(m^2 \log n)$-time $3/2$-approximation algorithm based on {\sc Algorithm A}.
\end{remark}

When the agreement graph $\overline{G}$ consists of $k$ vertex-disjoint triangles such that
a vertex of the $i$-th triangle is adjacent to a vertex of the $(i+1)$-st triangle, for $i = 1, 2, \ldots, k-1$, and the maximum degree is $3$,
{\sc Algorithm B} could produce a path cover containing $k$ $2$-paths, while there is a Hamiltonian path in the graph.
This suggests that the approximation ratio $4/3$ is asymptotically tight.

\section{Approximating $F2 \mid G = K_\ell \cup K_{n-\ell}, p_{ij} \mid C_{\max}$}
%==================================================================================================
In this section, we present a $3/2$-approximation algorithm for the weakly NP-hard problem $F2 \mid G = K_\ell \cup K_{n-\ell}, p_{ij} \mid C_{\max}$
for arbitrary jobs with a conflict graph that is the union of two disjoint cliques.
Therefore, the agreement graph $\overline{G} = K_{\ell, n-\ell}$ is a complete bipartite graph.
Without loss of generality, let the job set of $K_\ell$ be $A = \{J_1, J_2, \ldots, J_\ell\}$ and
the job set of $K_{n-\ell}$ be $B = \{J_{\ell+1}, J_{\ell+2}, \ldots, J_n\}$.

%We use a simple algorithm to warm up.
%
%
%\subsection{A $3/2$-approximation}
%--------------------------------------------------------------------------------------------------
For the job set $A$, we merge all its jobs (in the sequential order with increasing indices) to become a single ``{\em aggregated}'' job denoted as $J_A$,
with its processing time on the machine $M_1$ being $P_A^1 = \sum_{j=1}^\ell p_{1j}$ and
its processing time on the machine $M_2$ being $P_A^2 = \sum_{j=1}^\ell p_{2j}$.
Likewise, for the job set $B$, we merge all its jobs (in the sequential order with increasing indices) to become a single aggregated job denoted as $J_B$,
with its two processing times being $P_B^1 = \sum_{j=\ell+1}^n p_{1j}$ and $P_B^2 = \sum_{j=\ell+1}^n p_{2j}$.
We now have an instance of the classical two-machine flow-shop scheduling problem consisting of only two aggregated jobs $J_A$ and $J_B$,
and we may apply Johnson's algorithm~\cite{Joh54} to obtain a schedule denoted as $\pi$.
From $\pi$ we obtain a schedule for the original instance of the problem $F2 \mid G = K_\ell \cup K_{n-\ell}, p_{ij} \mid C_{\max}$,
which is also denoted as $\pi$ as there is no major difference.
We call this algorithm as {\sc Algorithm C}.

\begin{theorem}
\label{thm31}
{\sc Algorithm C} is an $O(m)$-time $3/2$-approximation algorithm for the problem $F2 \mid G = K_\ell \cup K_{n-\ell}, p_{ij} \mid C_{\max}$,
where $m$ is the number of edges in the conflict graph $G$.
\end{theorem}
\begin{proof}
Firstly, we note that {\sc Algorithm C} needs to spend $O(m)$ time to recognize that the conflict graph is indeed the union of two disjoint cliques,
and subsequently composes the two aggregated jobs.
If the two job subsets $A$ and $B$ are given without the need of recognition, then composing the two aggregated jobs can be done in $O(n)$ time,
where $n$ is the number of given jobs.
Scheduling two jobs on the two-machine flow-shop is done in constant time, afterwards the schedule $\pi$ for the original $n$ jobs can be constructed in $O(n)$ time.

Let $C_{\max}^*$ and $C_{\max}^\pi$ denote the optimal makespan and the makespan of the schedule $\pi$ produced by {\sc Algorithm C}, respectively.
One clearly sees that
\begin{equation}
\label{eq6}
C_{\max}^* \ge \max\{P_A^1 + P_A^2, P_A^1 + P_B^1, P_B^1 + P_B^2, P_A^2 + P_B^2\},
\end{equation}
in which each sum represents the total processing time of jobs in $A$,
the total processing time of jobs on the machine $M_1$,
the total processing time of jobs in $B$,
and the total processing time of jobs on the machine $M_2$, respectively.

Assume without loss of generality that $P_A^1 \le P_B^1$.

If $P_A^1 \le P_A^2$, then $C_{\max}^{\pi} = P_A^1 + \max\{P_A^2, P_B^1\} + P_B^2 \le P_A^1 + C_{\max}^* \le \frac 32 C_{\max}^*$;

if $P_A^1 > P_A^2 > P_B^2$, then $C_{\max}^{\pi} = P_A^1 + \max\{P_A^2, P_B^1\} + P_B^2 \le C_{\max}^* + P_B^2 \le \frac 32 C_{\max}^*$;

if $P_A^1 > P_A^2$ and $P_A^2 \le P_B^2$, then $C_{\max}^{\pi} = P_B^1 + \max\{P_B^2, P_A^1\} + P_A^2 \le C_{\max}^* + P_A^2 \le \frac 32 C_{\max}^*$.

This proves the theorem.
\end{proof}

%\subsection{A $11/8$-approximation}
%--------------------------------------------------------------------------------------------------
In the schedule produced by {\sc Algorithm C}, one sees that when the jobs of $A$ are processed on the machine $M_1$,
the other machine $M_2$ is left idle.
This is certainly disadvantageous.
For instance, when the jobs are all unit jobs and $|A| = |B| = \frac 12 n$, the makespan of the produced schedule is $\frac 32 n$,
while the agreement graph is Hamiltonian and thus by Eq.~(\ref{eq5}) the optimal makespan is only $n + 1$.
This huge gap suggests that one could probably design a better approximation and we leave it as an open question.
%Indeed, in the following we present an improved approximation algorithm that
%determines a subset of $B$ to be processed on $M_2$ while some jobs of $A$ are processed on $M_1$.
%
%
%\begin{theorem}
%\label{thm32}
%The problem $F2 \mid G = K_\ell \cup K_{n-\ell}, p_{ij} \mid C_{\max}$ admits an $O(m)$-time $11/8$-approximation algorithm,
%where $m$ is the number of edges in the conflict graph.
%\end{theorem}
%%
%\begin{proof}
%We construct in the following a schedule $\pi$ in $O(m)$ time and show that its makespan $C_{\max}^{\pi} \le \frac {11}8 C_{\max}^*$.
%
%We assume the two job subsets $A$ and $B$ are given at the front, as otherwise one can spend $O(m)$ time to determine them;
%and $P_A^1, P_A^2, P_B^1, P_B^2$ are calculated as in the last subsection.
%Our discuss is for the following case specified in Eq.~(\ref{eq7}), while all the other cases can be symmetrically discussed.
%%
%\begin{equation}
%\label{eq7}
%P_A^1 \le P_A^2, \ P_B^1 \le P_B^2, \mbox{ and } P_A^1 \le P_B^1.
%\end{equation}
%%
%The makespan of the schedule produced by {\sc Algorithm C} is $P_A^1 + \max\{P_A^2, P_B^1\} + P_B^2$.
%\end{proof}

\section{Concluding remarks}
%==================================================================================================
In this paper, we investigated the approximation algorithms for the two-machine flow-shop scheduling problem with a conflict graph,
in particular two special cases of all unit jobs and of a conflict graph that is the union of two disjoint cliques,
that is, $F2 \mid G = (V, E), p_{ij} = 1 \mid C_{\max}$ and $F2 \mid G = K_\ell \cup K_{n-\ell}, p_{ij} \mid C_{\max}$.
For the first problem we studied the graph theoretical problem of finding a path cover with the minimum total number of $0$-paths and $1$-paths,
and presented a polynomial time exact algorithm.
This exact algorithm leads to a $4/3$-approximation algorithm for the problem $F2 \mid G = (V, E), p_{ij} = 1 \mid C_{\max}$.
We also showed that the performance ratio $4/3$ is asymptotically tight.
For the second problem $F2 \mid G = K_\ell \cup K_{n-\ell}, p_{ij} \mid C_{\max}$, we presented a $3/2$-approximation algorithm.

We conjecture that designing approximation algorithms for $F2 \mid G = (V, E), p_{ij} = 1 \mid C_{\max}$ with a performance ratio better than $4/3$ is challenging,
since one way or the other one has to deal with longer paths in a path cover or has to deal with the original {\sc Path Cover} problem.
Nevertheless, better approximation algorithms for $F2 \mid G = K_\ell \cup K_{n-\ell}, p_{ij} \mid C_{\max}$ can be expected.

\subparagraph*{Acknowledgements.}
%-------------------------------------------------------------------------------
G. Chen, Y. Chen and A. Zhang are supported by the NSFC Grants 11571252 and 11771114;
Y. Chen is also supported by the China Scholarship Council Grant No. 201508330054;
R. Goebel and G. Lin are supported by NSERC Canada;
G. Lin is also supported by the NSFC Grant No. 61672323;
L. Liu is supported by the Fundamental Research Funds for the Central Universities (Grant No. 20720160035) and by the China Scholarship Council Grant No. 201706315073.

%\appendix
%\section{Morbi eros magna}
%==================================================================================================

%%
%% Bibliography
%%

%% Either use bibtex (recommended), 

%\bibliography{mybibfile,../BiBTeX/scheduling,../BiBTeX/mypapers,../BiBTeX/general}

%% .. or use the thebibliography environment explicitly

\end{document}